\documentclass{article}

\usepackage{amsmath,amsfonts,amssymb,amsthm}
\usepackage{slashed}
\usepackage{xy}
\usepackage[breaklinks=true]{hyperref}
\usepackage{float}

\input xy
\xyoption{all}

\newcommand{\maps}{\colon}    
\newcommand{\R}{{\mathbb R}}  
\newcommand{\C}{{\mathbb C}}  
\renewcommand{\H}{{\mathbb H}}  
\renewcommand{\O}{{\mathbb O}}  
\newcommand{\K}{{\mathbb K}}  
\newcommand{\Z}{{\mathbb Z}}  

\renewcommand{\Re}{\mathrm{Re}} 
\renewcommand{\Im}{\mathrm{Im}} 
\newcommand{\Retr}{\Re \; \tr} 
\newcommand{\h}{\mathfrak{h}} 
\newcommand{\tr}{{\mathrm{tr}}} 

\newcommand{\SO}{{\rm SO}}    
\newcommand{\Spin}{{\rm Spin}}    

\newcommand{\g}{{\mathfrak{g}}}  
\newcommand{\tri}{\operatorname{{{\rm tri}}}} 

\newcommand{\End}{{\rm End}} 
\newcommand{\Hom}{{\rm Hom}} 
\newcommand{\Cliff}{{\rm Cliff}}    

\newcommand{\inclusion}{\hookrightarrow}
\newcommand{\iso}{\cong} 
\newcommand{\tensor}{\otimes} 

\newcommand{\half}{\frac{1}{2}} 
\newcommand{\fourth}{\frac{1}{4}} 

\newcommand{\psibar}{\overline{\psi}} 
\newcommand{\chibar}{\overline{\chibar}} 

\newcommand{\define}[1]{{\bf #1}}

\newtheorem{thm}{Theorem}

\newtheorem{prop}[thm]{Proposition}

        \newcommand{\be}{\begin{equation}}
        \newcommand{\ee}{\end{equation}}
        \newcommand{\ba}{\begin{eqnarray}}
        \newcommand{\ea}{\end{eqnarray}}
        \newcommand{\ban}{\begin{eqnarray*}}
        \newcommand{\ean}{\end{eqnarray*}}
        \newcommand{\barr}{\begin{array}}
        \newcommand{\earr}{\end{array}}

\title{Division Algebras and Supersymmetry I}

\author{John C.\ Baez and John Huerta \\
\\
Department of Mathematics \\
University of California \\
Riverside, CA 92521 USA 
}

\date{December 29, 2009}

\begin{document}
\maketitle

\begin{abstract}
\noindent
Supersymmetry is deeply related to division algebras.  For example,
nonabelian Yang--Mills fields minimally coupled to massless spinors are
supersymmetric if and only if the dimension of spacetime is 3, 4, 6 or
10.  The same is true for the Green--Schwarz superstring.  In both
cases, supersymmetry relies on the vanishing of a certain trilinear
expression involving a spinor field.  The reason for this, in turn, is
the existence of normed division algebras in dimensions two less,
namely 1, 2, 4 and 8: the real numbers, complex numbers, quaternions 
and octonions.  Here we provide a self-contained account of how this works.  
\end{abstract}

\section{Introduction} \label{sec:intro}

There is a deep relation between supersymmetry and the four normed division
algebras: the real numbers $\R$, the complex numbers $\C$, the 
quaternions $\H$, and the octonions $\O$.  This is visible in 
the study of superstrings, supermembranes, and supergravity, 
but perhaps most simply in supersymmetric Yang--Mills theory.
In any dimension, we may consider a Yang--Mills field coupled to
a massless spinor transforming in the adjoint representation of 
the gauge group.  These fields are described by this Lagrangian:
\[ L = 
-\fourth \langle F, F \rangle + 
\half \langle \psi, \slashed{D}_A \psi \rangle. \]
Here $A$ is a connection on a bundle with semisimple gauge group $G$,
$F$ is the curvature of $A$, $\psi$ is a $\g$-valued spinor field, and
$\slashed{D}_A$ is the covariant Dirac operator associated with $A$.
It is well-known that this theory is supersymmetric if and only if the
dimension of spacetime is $3,4,6,$ or $10$.  Our goal here is to present
a self-contained proof of the `if' part of this result, based on the theory 
of normed division algebras.  

This result goes back to the work of Brink, Schwarz, and
Sherk~\cite{BrinkSchwarzScherk} and others. The book by Green, Schwarz
and Witten~\cite{GreenSchwarzWitten} contains a standard proof based
on the properties of Clifford algebras in various dimensions.  But
Evans~\cite{Evans} has shown that the supersymmetry of $L$ in
dimension $n+2$ implies the existence of a normed division algebra of
dimension $n$.  Conversely, Kugo and Townsend~\cite{KugoTownsend}
showed how spinors in dimension 3, 4, 6, and 10 derive special
properties from the normed division algebras $\R$, $\C$, $\H$ and
$\O$.  They formulated a supersymmetric model in 6 dimensions using
the quaternions, $\H$. They also speculated about a similar formalism
in 10 dimensions using the octonions, $\O$.

Shortly after Kugo and Townsend's work, Sudbery~\cite{Sudbery} used
division algebras to construct vectors, spinors and Lorentz groups in
Minkowski spacetimes of dimensions 3, 4, 6, and 10.  He then refined
his construction with Chung~\cite{ChungSudbery}, and with
Manogue~\cite{ManogueSudbery} he used these ideas to give an
octonionic proof of the supersymmetry of the above Lagrangian in
dimension 10.  This proof was later simplified by Manogue, Dray and
Janesky~\cite{DrayJaneskyManogue}. In the meantime,
Schray~\cite{Schray} applied the same tools to the superparticle.

All this work has made it quite clear that normed division algebras
explain why the above theory is supersymmetric in dimensions 3, 4, 6,
and 10.  Technically, what we need to check for supersymmetry is that
$\delta L$ is a total divergence with respect to the supersymmetry
transformation
\begin{eqnarray*}
    \delta A    & = & \epsilon \cdot \psi \\
    \delta \psi & = & \textstyle{\half} F \epsilon
\end{eqnarray*}
for any constant spinor field $\epsilon$.  (We explain the notation
here later; we assume no prior understanding of supersymmetry or 
normed division algebras.)  A calculation that works in any dimension 
shows that
\[ \delta L = \tri \psi + \mbox{divergence} \]
where $\tri \psi$ is a certain expression depending in a trilinear 
way on $\psi$ and linearly on $\epsilon$.  

So, the marvelous fact that needs to be understood is that $\tri \psi
= 0$ in dimensions 3, 4, 6, and 10, thanks to special properties of the
normed division algebras $\R$, $\C$, $\H$ and $\O$.  Indeed, this fact is 
responsible for supersymmetry, not only for Yang--Mills fields in
these dimensions, but also for superstrings!  The same term $\tri
\psi$ shows up as the obstruction to supersymmetry in the
Green--Schwarz Lagrangian for classical superstrings
\cite{GreenSchwarz,GreenSchwarzWitten}.  So, the vanishing of this
term deserves to be understood: clearly, simply, and in as many ways
as possible.

Unfortunately, many important pieces of the story are scattered
throughout the literature.  The treatment of Deligne and Freed
\cite{Deligne} is self-contained, and it uses normed division
algebras, but it does not use `purely equational reasoning': it proves
$\tri \psi = 0$ by first showing that the double cover of the Lorentz
group acts transitively on the set of nonzero spinors in dimensions 3,
4, 6, and 10.  While this geometrical argument is beautiful and
insightful, a purely equational approach has its own charm.  The line
of work carried out by Fairlie, Manogue, Sudbery, Dray, and collaborators
\cite{DrayJaneskyManogue, FairlieManogue, ManogueSudbery, Schray} has
shown that the equation $\tri \psi = 0$ can be derived from the
complete antisymmetry of another trilinear expression, the
`associator'
\[               [a,b,c] = (ab)c - a(bc)   \]
in the normed division algebra.  Our desire here is to merely present 
this argument as clearly as we can.

So, here we present an equational proof that $\tri \psi = 0$ in
dimensions 3, 4, 6, and 10, based on the complete antisymmetry of the
associator for the normed division algebras $\K = \R$, $\C$, $\H$ and 
$\O$.  In Section \ref{sec:divalg} we review the properties of normed
division algebras that we will need.  In Section
\ref{sec:intertwiners} we start by recalling how to interpret vectors
as $2 \times 2$ hermitian matrices with entries in $\K$, and spinors
as elements of $\K^2$.  We then use this language to describe the
basic operations involving vectors, spinors and scalars.
These include an operation that takes two spinors $\psi$ and $\phi$
and forms a vector $\psi \cdot \phi$, and an operation that takes a
vector $A$ and a spinor $\psi$ and forms a spinor $A \psi$.
In Section \ref{sec:fundamental} we prove the fundamental identity
that holds only in Minkowski spaces of dimensions 3, 4, 6 and 10:
\[               (\psi \cdot \psi) \psi = 0 . \]
Following Schray \cite{Schray}, we call this the `3-$\psi$'s rule'.  In 
Section \ref{sec:superalgebra} we introduce a little superalgebra, and 
explain why we should treat $\K$ as an `odd', or `fermionic', super vector
space.  In Section \ref{sec:sym} we formulate pure super-Yang--Mills
theory in terms of normed division algebras, completely avoiding the
use of gamma matrices.  We explain how the term $\tri \psi$ arises as
the obstruction to supersymmetry in this theory.  Finally, we use the
3-$\psi$'s rule to prove that $\tri \psi = 0$ in dimensions 3,
4, 6 and 10.

\section{Normed Division Algebras} \label{sec:divalg}

By a classic theorem of Hurwitz~\cite{Hurwitz}, there are only four
normed division algebras: the real numbers, $\R$, the complex numbers,
$\C$, the quaternions, $\H$, and the octonions, $\O$.  These algebras
have dimension 1, 2, 4, and 8.  For an overview of this subject,
including a Clifford algebra proof of Hurwitz's theorem, see
\cite{Baez:Octonions}.  Here we introduce the bare minimum of material
needed to reach our goal.

A \define{normed division algebra} $\K$ is a (finite-dimensional,
possibly nonassociative) real algebra equipped with a multiplicative
unit 1 and a norm $| \cdot |$ satisfying:
\[ |ab| = |a| |b|  \]
for all $a, b \in \K$.  Note this implies that $\K$ has no zero divisors.
We will freely identify $\R 1 \subseteq \K$ with $\R$.

In all cases, this norm can be defined using conjugation. Every normed
division algebra has a \define{conjugation} operator---a linear
operator $* \maps \K \to \K$ satisfying
\[ a^{**} = a, \quad (ab)^* = b^* a^* \]
for all $a,b \in \K$.   Conjugation lets us decompose each element of
$\K$ into real and imaginary parts, as follows:
\[ \Re(a) = \frac{a + a^*}{2}, \quad \Im(a) = \frac{a - a^*}{2}. \]
Conjugating changes the sign of the imaginary part and leaves the real part
fixed. We can write the norm as
\[ |a| = \sqrt{a a^*} = \sqrt{a^* a}. \]
This norm can be polarized to give an inner product on $\K$:
\[ (a, b) = \Re(a b^*) = \Re(a^* b). \]

The algebras $\R$, $\C$ and $\H$ are associative.  The octonions $\O$
are not.  Yet they come close: the subalgebra generated by any two
octonions is associative.  Another way to express this fact uses the
\define{associator}:
\[ [ a, b, c ] = (ab)c - a(bc), \]
a trilinear map $\K \tensor \K \tensor \K \to \K$.  A theorem due to
Artin \cite{Schafer} states that for any algebra, the subalgebra
generated by any two elements is associative if and only if the
associator is alternating (that is, completely antisymmetric in its
three arguments).  An algebra with this property is thus called
\define{alternative}.  The octonions $\O$ are alternative, and 
so of course are $\R$, $\C$ and $\H$: for these three the associator
simply vanishes!  

In what follows, our calculations make heavy use of the fact that
all four normed division algebras are alternative.  Besides this, the
properties we require are:

\begin{prop}
The associator changes sign when one of its entries is conjugated.
\end{prop}

\begin{proof}
Since the subalgebra generated by any two elements is associative, and
real elements of $\K$ lie in every subalgebra, $[a,b,c] = 0$ if any
one of $a,b,c$ is real.  It follows that $[a, b, c] = 
[ \Im(a), \Im(b), \Im(c)]$, which yields the desired result.
\end{proof}

\begin{prop}
\label{prop:imaginaryassociator} 
The associator is purely imaginary. 
\end{prop}

\begin{proof}
Since $(ab)^* = b^* a^*$, a calculation shows $[a,b,c]^* = -[c^*,b^*,a^*]$.  
By alternativity this equals $[a^*,b^*,c^*]$, which in
turn equals $-[a,b,c]$ by the above proposition.  So, $[a,b,c]$ is
purely imaginary.
\end{proof}

For any square matrix $A$ with entries in $\K$, we define its
\define{trace} $\tr(A)$ to be the sum of its diagonal entries. This
trace lacks the usual cyclic property, because $\K$ is noncommutative,
so in general $\tr(AB) \neq \tr(BA)$.  Luckily, taking the real part
restores this property:

\begin{prop}
Let $a$, $b$, and $c$ be elements of $\K$. Then
\[ \Re((ab)c) = \Re(a(bc)) \]
and this quantity is invariant under cyclic permutations of $a$, $b$, and
$c$.
\end{prop}

\begin{proof} 
Proposition~\ref{prop:imaginaryassociator} implies
that $\Re((ab)c) = \Re(a(bc))$. For the cyclic property, it 
then suffices to prove $\Re(ab) = \Re(ba)$.  Since $(a,b) = (b,a)$ and 
the inner product is defined by $(a, b) = \Re(ab^*) = \Re(a^* b)$, we see:
\[ \Re(ab^*) = \Re(b^*a). \]
The desired result follows upon substituting $b^*$ for $b$. 
\end{proof}

\begin{prop}
\label{prop:realtrace}
Let $A$, $B$, and $C$ be $k \times \ell$, $\ell \times m$ and
$m \times k$ matrices with entries in $\K$. Then
\[ \Retr((AB)C) = \Retr(A(BC)) \]
and this quantity is invariant under cyclic permutations of $A$, $B$,
and $C$.  We call this quantity the \define{real trace} $\Retr(ABC)$.
\end{prop}

\begin{proof}
This follows from the previous proposition and the definition of 
the trace.
\end{proof}

The reader will have noticed three trilinears in this section: the
associator $[a, b, c]$, the real part $\Re((ab)c)$, and the real trace
$\Retr(ABC)$.  This is no coincidence, as they all relate to the star
of the show, $\tri \psi$. In fact:
\[ \tri \psi = \Retr(\psi^\dagger (\epsilon \cdot \psi) \psi). \]
for some suitable matrices $\psi^\dagger$, $\epsilon \cdot \psi$ and
$\psi$.  Of course, we have not yet said how to construct these. We
turn to this now.

\section{Vectors, Spinors and Intertwiners} \label{sec:intertwiners}

It is well-known \cite{Baez:Octonions, KugoTownsend, Sudbery} that given a
normed division algebra $\K$ of dimension $n$, one can construct
$(n+2)$-dimensional Minkowski spacetime as the space of $2 \times 2$ hermitian
matrices with entries in $\K$, with the determinant giving the Minkowski
metric.  Spinors can then be described as elements of $\K^2$.  Our goal here is
to provide self-contained proofs of these facts, and then develop all the basic
operations involving vectors, spinors and scalars using this language.

To begin, let $\K[m]$ denote the space of $m \times m$ matrices with
entries in $\K$. Given $A \in \K[m]$, define its \define{hermitian adjoint}
$A^\dagger$ to be its conjugate transpose:
\[ A^\dagger = (A^*)^T. \]
We say such a matrix is \define{hermitian} if $A = A^\dagger$. Now take the
$2 \times 2$ hermitian matrices:
\[ 
\h_2(\K) = \left\{ 
\left( 
\begin{array}{c c} 
	t + x & y     \\
	y^*   & t - x \\
\end{array}
\right) 
\; : \;
t, x \in \R, \; y \in \K
\right\}.
\]
This is an $(n + 2)$-dimensional real vector space. Moreover, the
usual formula for the determinant of a matrix gives the Minkowski norm
on this vector space:
\[ 
-\det 
\left( 
\begin{array}{c c} 
	t + x & y     \\
	y^*   & t - x \\
\end{array}
\right) 
= - t^2 + x^2 + |y|^2.
\]
We insert a minus sign to obtain the signature $(n+1, 1)$. Note this
formula is unambiguous even if $\K$ is noncommutative or nonassociative.   

It follows that $\Spin(n + 1, 1)$, the double cover of the Lorentz group
$\SO_0(n+1,1)$, acts on $\h_2(\K)$ via determinant-preserving linear transformations.
Since this is the `vector' representation, we will often call $\h_2(\K)$ simply $V$.
The Minkowski metric
\[           g \maps V \otimes V \to \R  \]
is given by
\[           g(A,A) = -\det(A)  .\]
There is also a nice formula for the inner product of two different
vectors.  This involves the \define{trace reversal} of $A \in \h_2(\K)$,
introduced by Schray \cite{Schray} and defined as follows:
\[ \tilde{A} = A - (\tr A) 1. \]
Note we indeed have $\tr(\tilde{A}) = -\tr(A)$.  Also note that
\[ A =
\left( 
\begin{array}{c c} 
	t + x & y     \\
	y^*   & t - x \\
\end{array}
\right)
\qquad \implies \qquad \tilde{A} =
\left( 
\begin{array}{c c} 
	-t + x & y     \\
	y^*    & -t - x \\
\end{array}
\right)
\]
so trace reversal is really \emph{time reversal}.  Moreover:

\begin{prop}
\label{prop:metric}
For any vectors $A,B \in V = \h_2(K)$, we have
\[       A \tilde{A} = \tilde{A} A = - \det(A) 1 \]
and
\[  \frac{1}{2} \Retr(A \tilde{B}) = \frac{1}{2} \Retr(\tilde{A} B) = g(A,B)
\]
\end{prop}

\begin{proof}
We check the first equation by a quick calculation.  Taking the real
trace and dividing by 2 gives
\[      \frac{1}{2} \Retr(A \tilde{A}) = \frac{1}{2}
\Retr(\tilde{A}A) = -\det(A) = g(A,A). \]
Then we use the polarization identity, which says that two symmetric
bilinear forms that give the same quadratic form must be equal.
\end{proof}

Next we consider spinors.  As real vector spaces, the spinor
representations $S_+$ and $S_-$ are both just $\K^2$.
However, they differ as representations of $\Spin(n+1, 1)$. To
construct these representations, we begin by defining ways for vectors
to act on spinors:
\[ \begin{array}{cccl} 
	\gamma \maps & V \tensor S_+ & \to     & S_- \\
	             & A \tensor \psi     & \mapsto & A\psi.
   \end{array} \]
and 
\[ \begin{array}{cccl} 
	\tilde{\gamma} \maps & V \tensor S_- & \to     & S_+ \\
	                     & A \tensor \psi     & \mapsto & \tilde{A}\psi.
\end{array} \]
We can also think of these as maps that send elements of $V$ to linear 
operators:
\[ \begin{array}{cccl}
	\gamma \maps         & V & \to & \Hom(S_+, S_-), \\
	\tilde{\gamma} \maps & V & \to & \Hom(S_-, S_+).
\end{array} \]
Here a word of caution is needed: since $\K$ may be nonassociative,
$2 \times 2$ matrices with entries in $\K$ cannot be identified with 
linear operators on $\K^2$ in the usual way.  They
certainly induce linear operators via left multiplication:
\[ L_A(\psi) = A \psi. \]
Indeed, this is how $\gamma$ and $\tilde{\gamma}$ turn elements of $V$ into
linear operators:
\[ \begin{array}{ccl} 
	\gamma(A)         & = & L_A, \\
	\tilde{\gamma}(A) & = & L_{\tilde{A}}.
\end{array} \]
However, because of nonassociativity, composing such linear operators is
different from multiplying the matrices:
\[ L_A L_B(\psi) = A( B \psi) \neq (AB) \psi = L_{AB} (\psi). \]

Since vectors act on elements of $S_+$ to give elements of $S_-$
and vice versa, they map the space $S_+ \oplus S_-$ to itself.
This gives rise to an action of the Clifford algebra 
$\Cliff(V)$ on $S_+ \oplus S_-$:

\begin{prop}
The vectors $V = \h_2(\K)$ act on the spinors $S_+ \oplus S_- = \K^2
\oplus \K^2$ via the map
\[ \Gamma \maps  V \to      \End(S_+ \oplus S_-) \]
given by 
\[   \Gamma(A)(\psi, \,\phi) = (\widetilde{A} \phi, \, A \psi)  .\]
Furthermore, $\Gamma(A)$ satisfies the Clifford algebra relation:
\[ \Gamma(A)^2 = g(A,A) 1 \]
and so extends to a homomorphism $\Gamma \maps \Cliff(V) \to 
\End(S_+ \oplus S_-)$, i.e.\ a representation of the Clifford algebra 
$\Cliff(V)$ on $S_+ \oplus S_-$.
\end{prop}

\begin{proof}
Suppose $A \in V$ and $\Psi = (\psi, \phi) \in S_+ \oplus S_-$. 
We need to check that 
\[ \Gamma(A)^2(\Psi) = -\det(A) \Psi. \]
Here we must be mindful of nonassociativity: we have
\[ \Gamma(A)^2(\Psi) = ( \tilde{A}(A \psi), \, A(\tilde{A} \phi)) . \]
Yet it is easy to check that the expressions $\tilde{A}(A \psi)$ 
and $A (\tilde{A} \phi)$ involve multiplying at most two 
different nonreal elements of $\K$.
These associate, since $\K$ is alternative, so in fact
\[ \Gamma(A)^2(\Psi) = ( (\tilde{A}A) \psi, \, (A \tilde{A}) \phi) . \]
To conclude, we use Proposition \ref{prop:metric}.
\end{proof}

The action of a vector swaps $S_+$ and $S_-$, so acting by vectors twice sends $S_+$
to itself and $S_-$ to itself.  This means that while $S_+$ and $S_-$ are \emph{not}
modules for the Clifford algebra $\Cliff(V)$, they are both modules for the even part
of the Clifford algebra, generated by products of pairs of vectors. The group
$\Spin(n+1,1)$ lives in this even part. Indeed, call a vector $A$ such that $g(A,A)
= \pm 1$ a \define{unit vector}.  It is well known that the group in $\Cliff_0(V)$ generated by products of pairs of unit vectors is a double cover of $\SO(n+1,1)$,  and thus
its identity component is the double cover of $\SO_0(n+1,1)$. This identity component
is therefore $\Spin(n+1,1)$.

While we will not need this in what follows, one can check that:
\begin{itemize}
	\item When $\K = \R$, $S_+ \iso S_-$ is the Majorana spinor
	representation of $\Spin(2,1)$.
	\item When $\K = \C$, $S_+ \iso S_-$ is the Majorana spinor
	representation of $\Spin(3,1)$.
	\item When $\K = \H$, $S_+$ and $S_-$ are the Weyl spinor
	representations of $\Spin(5,1)$.
	\item When $\K = \O$, $S_+$ and $S_-$ are the Majorana--Weyl
	spinor representations of $\Spin(9,1)$.
\end{itemize}
This counts as a consistency check, because these are precisely the kinds of
spinor representations that go into pure super-Yang--Mills theory.  But it is
important to note that the \emph{differences} between these spinor
representations are irrelevant to our argument. What matters is how they
are the \emph{same}---they can all be defined on $\K^2$.

Now that we have representations of $\Spin(n+1, 1)$ on $V$, $S_+$ and
$S_-$, we need to develop the $\Spin(n+1, 1)$-equivariant maps that
relate them.  Ultimately, to define the Lagrangian for pure
super-Yang--Mills theory, we need:
\begin{itemize}
	\item An invariant pairing: 
		\[ \langle -, - \rangle \maps S_+ \tensor S_- \to \R. \]
	\item An equivariant map that turns pairs of spinors into vectors:
		\[ \cdot \, \maps S_\pm \tensor S_\pm \to V. \]
\end{itemize}
Another name for an equivariant map between group representations
is an `intertwining operator'.  As a first step, we show that the 
action of vectors on spinors is itself an intertwining operator:

\begin{prop} 
The maps 
\[ \begin{array}{cccl}
	\gamma \maps & V \tensor S_+ & \to     & S_- \\
                     & A \tensor \psi     & \mapsto & A \psi
\end{array} \]
and
\[ \begin{array}{cccl}
	\tilde{\gamma} \maps & V \tensor S_- & \to     & S_+ \\
	                     & A \tensor \psi    & \mapsto & \tilde{A} \psi
\end{array} \]
are equivariant with respect to the action of $\Spin(n+1, 1)$.
\end{prop}

\begin{proof}
Both $\gamma$ and $\tilde{\gamma}$ are restrictions of the map \[ \Gamma \maps
V \tensor (S_+ \oplus S_-) \to S_+ \oplus S_- ,\] so it suffices to check that
$\Gamma$ is equivariant. Indeed, an element $g \in \Spin(n+1,1)$ acts on $V$ by
conjugation on $V \subseteq \Cliff(V)$, and it acts on $S_+ \oplus S_-$ by
$\Gamma(g)$. Thus, we compute:
\[ \Gamma(gAg^{-1}) \Gamma(g)\Psi = \Gamma(g) (\Gamma(A) \Psi), \]
for any $\Psi \in S_+ \oplus S_-$. Here it is important to note that the
conjugation $gAg^{-1}$ is taking place in the associative algebra $\Cliff(V)$,
not in the algebra of matrices.  This equation says that $\Gamma$ is indeed
$\Spin(n+1,1)$-equivariant, as claimed.  
\end{proof}

Now we exhibit the key tool: the pairing between $S_+$
and $S_-$:

\begin{prop}
The pairing
\[ \begin{array}{cccl}
\langle -, - \rangle \maps & S_+ \tensor S_- & \to     & \R \\
 & \psi \tensor \phi & \mapsto & \Re(\psi^\dagger \phi)
\end{array} \]
is invariant under the action of $\Spin(n+1,1)$.
\end{prop}

\begin{proof}
Given $A \in V$, we use the fact that the associator is purely
imaginary to show that 
\[ 
\Re \left( ( \tilde{A} \phi )^\dagger ( A \psi ) \right) = 
\Re \left( (\phi^\dagger \tilde{A}) (A \psi) \right) = 
\Re \left( \phi^\dagger (\tilde{A} (A \psi)) \right).
\]
As in the proof of the Clifford relation, it is easy to check that the column
vector $\tilde{A}(A \psi)$ involves at most two nonreal elements of $\K$ and
equals $g(A,A) \psi$. So:
\[ 
\langle \tilde{\gamma}(A) \phi, \gamma(A) \psi \rangle = 
g(A,A) \langle \psi, \phi \rangle.
 \]
In particular when $A$ is a unit vector, acting by $A$ swaps the order of $\psi$ and
$\phi$ and changes the sign at most. In fact, this implies our result, though we
need a more explicit presentation of $\Spin(n+1,1)$ to see this. Proposition 5.4.8
of Varadarajan \cite{Varadarajan} tells us that $\Spin(n+1,1)$ is the group
generated by products of even numbers of unit vectors, an even number of which
satisfy $g(A,A) = -1$:
\[ \Spin(n+1,1) = \left\langle A_1 \cdots A_{2p} B_1 \cdots B_{2q} \, : \, A_i, B_j \in V, \, g(A_i,A_i) = 1, \, g(B_j, B_j) = -1 \right\rangle . \]
By the computation above, this implies that $\langle -, - \rangle$ is invariant under
$\Spin(n+1,1)$.
\end{proof}

With this pairing in hand, there is a manifestly equivariant way to
turn a pair of spinors into a vector.  Given 
$\psi, \phi \in S_+$, there is a unique vector $\psi \cdot \phi$ 
whose inner product with any vector $A$ is given by
\[ g(\psi \cdot \phi, A) = \langle \psi, \gamma(A) \phi \rangle .\]
Similarly, given $\psi, \phi \in S_-$, we define 
$\psi \cdot \phi \in V$ by demanding
\[ g(\psi \cdot \phi, A) = \langle \tilde{\gamma}(A) \psi, \phi \rangle \]
for all $A \in V$.  This gives us maps
\[ S_\pm \tensor S_\pm \to V \]
which are manifestly equivariant.

On the other hand, because $S_\pm = \K^2$ and $V = \h_2(\K)$, there is
also a naive way to turn a pair of spinors into a vector using matrix
operations: just multiply the column vector $\psi$ by the row vector
$\phi^\dagger$ and then take the hermitian part:
\[ \psi \phi^\dagger + \phi \psi^\dagger \in \h_2(\K), \]
or perhaps its trace reversal:
\[ \widetilde{\psi \phi^\dagger + \phi \psi^\dagger} \in \h_2(\K). \]
In fact, these naive guesses match the manifestly equivariant approach
described above:

\begin{prop}
The maps $\cdot \, \maps S_\pm \tensor S_\pm \to V$ are given by:
\[ \begin{array}{cccl}
	\cdot \, \maps & S_+ \tensor S_+ & \to     & V \\
 & \psi \tensor \phi   & \mapsto & 
\widetilde{\psi \phi^\dagger + \phi \psi^\dagger}
\end{array} \]
\[ \begin{array}{cccl}
	\cdot \, \maps & S_- \tensor S_- & \to     & V \\
& \psi \tensor \phi  & \mapsto & 
\psi \phi^\dagger + \phi \psi^\dagger.
\end{array} \]
These maps are equivariant with respect to the action of $\Spin(n+1, 1)$.
\end{prop}

\begin{proof}
First suppose $\psi,\phi \in S_+$.  We have already seen that the map
$\cdot \maps S_+ \tensor S_+ \to V$ is equivariant.
We only need to show that this map has the desired form.
We start by using some definitions:
\[     g(\psi \cdot \phi, A) =
	\langle \psi, \gamma(A) \phi \rangle =
       \Re(\psi^\dagger (A \phi)) =
       \Retr( \psi^\dagger A \phi). \]
We thus have
\[     g(\psi \cdot \phi, A) =
       \Retr( \psi^\dagger A \phi) =
	\Retr( \phi^\dagger A \psi),
\]
where in the last step we took the adjoint of the inside. 
Applying the cyclic property of the real trace, we obtain
\[ 
g(\psi \cdot \phi, A) = 
\Retr( \phi \psi^\dagger A ) = 
\Retr( \psi \phi^\dagger A ). 
\]
Averaging gives
\[ 
g(\psi \cdot \phi, A) =
\half \Retr( ( \psi \phi^\dagger + \phi \psi^\dagger ) A ). 
\]
On the other hand, Proposition \ref{prop:metric} implies that
\[ g(\psi \cdot \phi, A) 
= \frac{1}{2} \Retr(\widetilde{(\psi \cdot \phi)}A) .\]
Since both these equations hold for all $A$, we must have 
\[   \widetilde{\psi \cdot \phi} =  \psi \phi^\dagger + \phi \psi^\dagger. \]
Doing trace reversal twice gets us back where we started, so
\[   \psi \cdot \phi = \widetilde{\psi \phi^\dagger + \phi \psi^\dagger}
\]
as desired.  A similar calculation shows that if $\psi, \phi \in S_-$, then
$\psi \cdot \phi = \psi \phi^\dagger + \phi \psi^\dagger$.
\end{proof}

\begin{table}[H]
	\begin{center}
\renewcommand{\arraystretch}{1.3}
		\begin{tabular}{r@{$\maps$}c@{$\tensor$}c@{$\to$}ccc}
			\hline
			\multicolumn{4}{c}{Map}                                  & Division algebra notation                           & Index notation\\
			\hline   
			$g$                     & $\, V$   & $\, V$   & $\, \R$  & $\frac{1}{2} \Retr(A \tilde{B})$                    & $A^\mu B_\mu$ \\
			$\gamma$                & $\, V$   & $\, S_+$ & $\, S_-$ & $A\psi$                                             & $\gamma_\mu A^\mu \psi$ \\
			$\tilde{\gamma}$        & $\, V$   & $\, S_-$ & $\, S_+$ & $\tilde{A}\psi$                                     & $\tilde{\gamma}_\mu A^\mu \psi$ \\
			$\cdot \,$              & $\, S_+$ & $\, S_+$ & $\, V$   & $\widetilde{\psi \phi^\dagger + \phi \psi^\dagger}$ & $\psibar \gamma^\mu \phi$ \\
			$\cdot \,$              & $\, S_-$ & $\, S_-$ & $\, V$   & $\psi \phi^\dagger + \phi \psi^\dagger$             & $\psibar \tilde{\gamma}^\mu \phi$ \\
			$\langle - , - \rangle$ & $\, S_+$ & $\, S_-$ & $\, \R$  & $\Re(\psi^\dagger \phi)$                            & $\psibar \phi$ \\
			\hline
		\end{tabular}
\caption{\label{tab:intertwiners}
Division algebra notation vs.\ index notation}
	\end{center}
\end{table}
\renewcommand{\arraystretch}{1}
\noindent 
We can summarize our work so far with a table of the basic bilinear
maps involving vectors, spinors and scalars.  Table 1 shows how to
translate between division algebra notation and something more closely
resembling standard physics notation.  In this table the adjoint
spinor $\psibar$ denotes the spinor dual to $\psi$ under the pairing
$\langle -, - \rangle$.  The gamma matrix $\gamma^\mu$ denotes a
Clifford algebra generator acting on $S_+$, while $\tilde{\gamma}^\mu$
denotes the same element acting on $S_-$.  Of course $\tilde{\gamma}$
is not standard physics notation; the standard notation for this
depends on which of the four cases we are considering: $\R$, $\C$, $\H$ or 
$\O$.

\section{The 3-$\psi$'s Rule} \label{sec:fundamental}

Now we prove the fundamental identity that makes supersymmetry tick in
dimensions 3, 4, 6, and 10.  This identity was dubbed the `3-$\psi$'s
rule' by Schray \cite{Schray}.  The following proof is based on an
argument in the appendix of the paper by Dray, Janesky and Manogue
\cite{DrayJaneskyManogue}.  Note that it is really the alternative law,
rather than the normed division algebra axioms, that does the job:

\begin{thm}  
\label{thm:fundamental_identity}
Suppose $\psi \in S_+$.  Then $(\psi \cdot \psi) \psi = 0$. Similarly, if 
$\phi \in S_-$, then $(\widetilde{\phi \cdot \phi}) \phi = 0$.
\end{thm}

\begin{proof}
Suppose $\psi \in S_+$.  By definition,
\[     (\psi \cdot \psi) \psi =  2(\widetilde{\psi \psi^\dagger}) \psi 
= 2(\psi \psi^\dagger - \tr(\psi \psi^\dagger) 1) \psi .\]
It is easy to check that $\tr (\psi \psi^\dagger) = \psi^\dagger \psi$,
so 
\[     (\psi \cdot \psi) \psi =  
2((\psi \psi^\dagger) \psi - (\psi^\dagger \psi) \psi ).\]
Since $\psi^\dagger \psi$ is a real number, it commutes with $\psi$:
\[     (\psi \cdot \psi) \psi =  
2((\psi \psi^\dagger) \psi - \psi (\psi^\dagger \psi) ) .\]
Since $\K$ is alternative, every subalgebra of $\K$ generated by two elements
is associative.  Since $\psi \in \K^2$ is built from
just two elements of $\K$, the right-hand side vanishes.
The proof of the identity for $\phi \in S_-$ is similar. 
\end{proof}

It will be useful to state this result in a somewhat more elaborate
form.  To save space we only give this version for spinors in
$S_+$, though an analogous result holds for spinors in $S_-$:

\begin{thm}
\label{thm:cubic}
Define a map
\[ \begin{array}{cccl}
	T \maps & S_+ \tensor S_+ \tensor S_+ & \to     & S_- \\
       & \psi \tensor \phi \tensor \chi     
      & \mapsto & 
(\psi \cdot \phi) \chi + (\phi \cdot \chi) \psi + (\chi \cdot \psi) \phi. 
\end{array} \]
Then $T = 0$.
\end{thm}

\begin{proof}
It is easy to check that $\psi \cdot \phi = \phi \cdot \psi$ for all
$\psi, \phi \in S_+$, so the map $T$ is completely symmetric in its
three arguments.  Just as any symmetric bilinear form $B(x,y)$ can be
recovered from the corresponding quadratic form $B(x,x)$ by
polarization, so too can any symmetric trilinear form be recovered
from the corresponding cubic form.  Since $T(\psi, \psi, \psi) = 0$ by
Theorem \ref{thm:fundamental_identity}, it follows that $T = 0$.
\end{proof}

To see \emph{how} this theorem is the key to supersymmetry for
super-Yang--Mills theory, we need a little superalgebra.

\section{Superalgebra} \label{sec:superalgebra}

So far we have used normed division algebras to construct a number of
algebraic structures: vectors as elements of $\h_2(\K)$, spinors as
elements of $\K^2$, and the various bilinear maps involving vectors,
spinors, and scalars.  However, to describe supersymmetry, we also
need superalgebra.  Specifically, we need anticommuting spinors.
Physically, this is because spinors are fermions, so we need them to
satisfy anticommutation relations.  Mathematically, this means that we
will do our algebra in the category of `super vector spaces',
SuperVect, rather than the category of vector spaces, Vect.

A \define{super vector space} is a $\Z_2$-graded vector space $V = V_0
\oplus V_1$ where $V_0$ is called the \define{even} or
\define{bosonic} part, and $V_1$ is called the \define{odd} or
\define{fermionic} part.  Like Vect, SuperVect is a symmetric monoidal 
category \cite{Baez:Rosetta}.  It has:
\begin{itemize}
	\item $\Z_2$-graded vector spaces as objects;
	\item Grade-preserving linear maps as morphisms;
	\item A tensor product $\tensor$ that has the following grading: if $V
		= V_0 \oplus V_1$ and $W = W_0 \oplus W_1$, then $(V \tensor
		W)_0 = (V_0 \tensor W_0) \oplus (V_1 \tensor W_1)$ and 
              $(V \tensor W)_1 = (V_0 \tensor W_1) \oplus (V_1 \tensor W_0)$;
	\item A braiding
		\[ B_{V,W} \maps V \tensor W \to W \tensor V \]
		defined as follows: $v \in V$ and $w \in W$ 
              are of grade $p$ and $q$, then
		\[ B_{V,W}(v \tensor w) = (-1)^{pq} w \tensor v. \]
\end{itemize}
The braiding encodes the `the rule of signs': in any calculation, when
two odd elements are interchanged, we introduce a minus sign.  

In what follows \emph{we treat the normed division algebra $\K$ as an odd
super vector space}.   This turns out to force the spinor representations
$S_\pm$ to be odd and the vector representation $V$ to be even, as follows.

There is an obvious notion of direct sums for super vector spaces, with
\[         (V \oplus W)_0 = V_0 \oplus W_0  , \qquad 
           (V \oplus W)_1 = V_1 \oplus W_1 \]
and also an obvious notion of duals, with
\[         (V^*)_0 = (V_0)^*, \qquad (V^*)_1 = (V_1)^*  .\]
We say a super vector space $V$ is \define{even} if it equals its 
even part ($V = V_0$), and \define{odd} if it equals its odd part 
($V = V_1$).  Any subspace $U \subseteq V$ of an even (resp.\ odd) 
super vector space becomes a super vector space which is again even
(resp.\ odd).
                         
We treat the spinor representations $S_\pm$ as super vector spaces
using the fact that they are the direct sum of two copies of $\K$.
Since $\K$ is odd, so are $S_+$ and $S_-$.  Since $\K^2$ is odd, so is
its dual.  This in turn forces the space of linear maps from $\K^2$ to
itself, $\End(\K^2) = \K^2 \tensor (\K^2)^*$, to be even. This even
space contains the $2 \times 2$ matrices $\K[2]$ as the subspace of
maps realized by left multiplication:
\[ \begin{array}{rcl}
\K[2] & \inclusion & \End(\K^2) \\
A     & \mapsto         & L_A .
\end{array} \]
$\K[2]$ is thus even. Finally, this forces the subspace of hermitian
$2 \times 2$ matrices, $\h_2(\K)$, to be even.  So, the vector
representation $V$ is even.  All this matches the usual rules in
physics, where spinors are fermionic and vectors are bosonic.

\section{Super-Yang--Mills Theory} \label{sec:sym}

We are now ready to give a division algebra interpretation
of the pure super-Yang--Mills Lagrangian
\[ L = -\fourth \langle F, F \rangle + 
\half \langle \psi, \slashed{D}_A \psi \rangle \]
and use this to prove its supersymmetry. For simplicity, we shall work
over Minkowski spacetime, $M$.  This allows us to treat all bundles as
trivial, sections as functions, and connections as $\g$-valued 1-forms.

At the outset, we fix an invariant inner product on $\g$, the Lie algebra of
a semisimple Lie group $G$. We shall use the following standard tools from
differential geometry to construct $L$, none of which need involve spinors
or division algebra technology:
\begin{itemize}
	\item A connection $A$ on a principal $G$-bundle over $M$.
	      Since the bundle is trivial we think of this connection 
             as a $\g$-valued 1-form.
	\item The exterior covariant derivative $d_A = d + [A, -]$ on 
             $\g$-valued $p$-forms.
	\item The curvature $F = dA + \frac{1}{2} [A,A]$,  
             which is a $\g$-valued 2-form.
	\item The usual pointwise inner product $\langle F, F \rangle$ on
		$\g$-valued 2-forms, defined using the Minkowski metric on $M$
		and the invariant inner product on $\g$.
\end{itemize}
We also need the following spinorial tools. Recall from the preceding
section that $S_+$ and $S_-$ are odd objects in SuperVect.
So, whenever we switch two spinors, we introduce a minus sign.
\begin{itemize}
	\item A $\g$-valued section $\psi$ of a spin bundle over $M$. Note that
		this is, in fact, just a function:
		\[ \psi \maps M \to S_\pm \tensor \g. \]
		We call the collection of all such functions $\Gamma(S_\pm
		\tensor \g)$.
	\item The covariant Dirac operator $\slashed{D}_A$ derived from the
              connection $A$.  Of course,
		\[ \slashed{D}_A \maps \Gamma(S_\pm \tensor \g) \to \Gamma(S_\mp \tensor \g) \]
		and in fact,
		\[ \slashed{D}_A = \slashed{\partial} + A. \]
	\item A bilinear pairing 
		\[ \langle -,- \rangle \maps \Gamma(S_+ \tensor \g) \tensor \Gamma(S_- \tensor \g) \to C^\infty(M) \]
		built pointwise using our pairing
		\[ \langle -,- \rangle \maps S_+ \tensor S_- \to \R \]
		and the invariant inner product on $\g$.
\end{itemize}

The basic fields in our theory are a
connection on a principal $G$-bundle, which we think of as a
$\g$-valued 1-form:
\[ A \maps M \to V^* \tensor \g. \]
and a $\g$-valued spinor field, which we think of as a
$S_+ \tensor \g$-valued function on $M$:
\[ \psi \maps M \to S_+ \tensor \g .\]
All our arguments would work just as well with
$S_-$ replacing $S_+$.

To show that $L$ is supersymmetric, we need to show $\delta L$ is
a total divergence when $\delta$ is the following supersymmetry 
transformation:
\begin{eqnarray*}
    \delta A    & = & \epsilon \cdot \psi \\
    \delta \psi & = & \half F \epsilon
\end{eqnarray*}
where $\epsilon$ is an arbitrary constant spinor field, treated as 
odd, but not $\g$-valued.  By a \define{supersymmetry transformation} we mean
that computationally we treat $\delta$ as a derivation. So, it is linear: 
\[      \delta (\alpha f + \beta g) = \alpha \delta f + \beta \delta g \]
where $\alpha, \beta \in \R$, and it satisfies the product rule:
\[         \delta (f g) = \delta(f) g + f \delta g. \]
For a more formal definition of `supersymmetry transformation' see
\cite{Deligne}.

The above equations require further explanation.  The dot in 
$\epsilon \cdot \psi$ denotes an operation that combines 
the spinor $\epsilon$ with the $\g$-valued spinor $\psi$ 
to produce a $\g$-valued 1-form.  We build this from our basic
intertwiner
\[ \cdot \, \maps S_+ \tensor S_+ \to V. \]
We identify $V$ with $V^*$ using the Minkowski inner product $g$, obtaining
\[ \cdot \, \maps S_+ \tensor S_+ \to V^*. \]
Then we tensor both sides with $\g$. This gives us a way to act by a 
spinor field on a $\g$-valued spinor field to obtain a $\g$-valued 1-form.
We take the liberty of also denoting this with a dot:
\[ \cdot \, \maps \Gamma(S_+) \tensor \Gamma(S_+ \tensor \g) 
\to \Omega^1(M,\g). \]

We also need to explain how the 2-form $F$ acts on the constant spinor 
field $\epsilon$.  Using the Minkowski metric, we can identify differential
forms on $M$ with sections of the Clifford algebra bundle over $M$:
\[ \Omega^* (M) \iso \Cliff(M). \]
Using this, differential forms act on spinor fields.  Tensoring
with $\g$, we obtain a way for $\g$-valued differential forms like $F$ 
to act on spinor fields like $\epsilon$ to give $\g$-valued spinor fields
like $F \epsilon$.  

Let us now apply the supersymmetry transformation
to each term in the Lagrangian.  First, the bosonic term:

\begin{prop}
The bosonic term has:
\[
 \delta \langle F, F \rangle 
= 2 (-1)^{n+1} \, \langle \psi, ({\star d_A \star} \, F) \epsilon \rangle + 
\rm{divergence}. 
\]
\end{prop}

\begin{proof}
By the symmetry of the inner product, we get:
\[ \delta \langle F, F \rangle = 2 \langle F, \delta F \rangle. \]
Using the handy formula $\delta F = d_A \delta A$, we have:
\[ \langle F, \delta F \rangle = \langle F, d_A \delta A \rangle. \]
Now the adjoint of the operator $d_A$ is $\star d_A \star$, up
to a pesky sign: if $\nu$ is a $\g$-valued $(p-1)$-form and $\mu$ is
a $\g$-valued $p$-form, we have 
\[   \langle \mu, d_A \nu \rangle =
(-1)^{dp + d + 1 + s} \langle {\star d_A \star} \, \mu, \nu \rangle 
+ \mbox{divergence} \]
where $d$ is the dimension of spacetime and $s$ is the signature,
i.e., the number of minus signs in the diagonalized metric.  It
follows that
\[ \langle F, \delta F \rangle = \langle F, d_A \delta A \rangle =
(-1)^n \,
\langle {\star d_A \star} \, F, \delta A \rangle + \mbox{divergence} \]
where $n$ is the dimension of $\K$.  
By the definition of $\delta A$, we get
\[ \langle {\star d_A \star} \, F, \delta A \rangle
=
\langle  {\star d_A \star} \, F, \epsilon \cdot \psi \rangle. 
\]
Now we can use division algebra technology to show:
\[
\langle {\star d_A \star} \, F , \epsilon \cdot \psi \rangle = \half
\Retr\left( ({\star d_A \star} \, F) (\epsilon \psi^\dagger + \psi
\epsilon^\dagger) \right) = -\langle \psi, ({\star d_A \star} \, F)
\epsilon \rangle, 
\] 
using the cyclic property of the real trace in the last step, and
introducing a minus sign in accordance with the sign rule.  Putting
everything together, we obtain the desired result.
\end{proof}

Even though this proposition involved the bosonic term only, division
algebra technology was still a useful tool in its proof. This is even
more true in the next proposition, which deals with the the fermionic term:

\begin{prop}
The fermionic term has:
\[ \delta \langle \psi, \slashed{D}_A \psi \rangle = \langle \psi, 
\slashed{D}_A(F \epsilon) \rangle + \tri \psi + \rm{divergence} \]
where
\[ \tri \psi = \langle \psi, (\epsilon \cdot \psi) \psi \rangle. \]
\end{prop}

\begin{proof}
It is easy to compute:
\[ \delta \langle \psi, \slashed{D}_A \psi \rangle = 
\langle \delta \psi, \slashed{D}_A \psi \rangle + 
\langle \psi, \delta{\slashed{D}_A} \psi \rangle + 
\langle \psi, \slashed{D}_A \delta \psi \rangle. \]
Now we insert $\delta \slashed{D}_A = \delta A = \epsilon \cdot \psi$, and 
thus see that the penultimate term is the trilinear one:
\[ \tri \psi = \langle \psi, (\epsilon \cdot \psi) \psi \rangle. \]
So, let us concern ourselves with the remaining terms:
\[ \langle \delta \psi, \slashed{D}_A \psi \rangle + 
\langle \psi, \slashed{D}_A \delta \psi \rangle. \]
A computation using the product rule shows that the divergence of the 
1-form $\psi \cdot \phi$ is given by $-\langle \phi, \slashed{D}_A \psi \rangle
+ \langle \psi, \slashed{D}_A \phi \rangle$, where the minus sign on the first
term arises from using the sign rule with these odd spinors. In the terms under
consideration, we can use this identity to move $\slashed{D}_A$ onto $\delta
\psi$:
\[ \langle \delta \psi, \slashed{D}_A \psi \rangle + \langle \psi, 
\slashed{D}_A \delta \psi \rangle = 
2 \langle \psi, \slashed{D}_A \delta \psi \rangle + \mbox{divergence}. \]
Substituting $\delta \psi = \half F \epsilon$, we obtain the desired
result.
\end{proof}

Using these two propositions, it is immediate that
\begin{eqnarray*}
	\delta L & = & 
-\fourth \delta \langle F,F \rangle + 
\half \delta \langle \psi, \slashed{D}_A \psi \rangle \\
        & = & \half (-1)^n 
\langle \psi, ({\star d_A \star}\, F) \epsilon \rangle + 
\half \langle \psi, \slashed{D}_A(F \epsilon) \rangle + 
\half \tri \psi + \rm{divergence} 
\end{eqnarray*}
All that remains to show is that $\slashed{D}_A(F \epsilon) = (-1)^{n+1}
({\star d_A \star}F) \, \epsilon$.  Indeed, Snygg shows (Eq.\ 7.6 
in~\cite{Snygg}) that for an
ordinary, non-$\g$-valued $p$-form $F$ 
\[ \slashed{\partial} (F \epsilon) =
(d F) \epsilon + (-1)^{d + dp + s} ({\star d \star} \, F) \epsilon  \] 
where $d$ is the dimension of spacetime and $s$ is the signature. 
This is easily generalized to
covariant derivatives and $\g$-valued $p$-forms:
\[ \slashed{D}_A (F \epsilon) = 
(d_A F) \epsilon + (-1)^{d + dp + s} ({\star d_A \star} \, F) \epsilon .\] 
In particular, when $F$ is the curvature 2-form, the first term vanishes by the
Bianchi identity $d_A F = 0$, and we are left with:
\[ \slashed{D}_A (F \epsilon) = (-1)^{n+1} ({\star d_A \star} \, F) \epsilon \]
where $n$ is the dimension of $\K$. We have thus shown:

\begin{prop}
\label{prop:variation}
Under supersymmetry transformations, the Lagrangian $L$ has: 
\[   \delta L = \half \tri \psi + \rm{divergence}.   \]
\end{prop}

The above result actually holds in every dimension, though our proof
used division algebras and was thus adapted to the dimensions of
interest: 3, 4, 6, and 10.  The next result is where division algebra
technology becomes really crucial:

\begin{prop}
\label{prop:variation2}
For Minkowski spacetimes of dimensions 3, 4, 6, and 10, 
$\tri \psi = 0$.  
\end{prop}

\begin{proof} 
At each point, we can write 
\[ \psi = \sum \psi^a \tensor g_a, \]
where $\psi^a \in S_+$ and $g_a \in \g$. When we insert this into $\tri
\psi$, we see that
\[ \tri \psi = 
\sum \langle \psi^a, (\epsilon \cdot \psi^b) \psi^c \rangle 
\, \langle g_a, [g_b, g_c] \rangle. \]
Since $\langle g_a, [g_b, g_c] \rangle$ is totally antisymmetric, this
implies $\tri \psi = 0$ for all $\epsilon$ if and only if the part of
$\langle \psi^a, (\epsilon \cdot \psi^b) \psi^c \rangle$ that is
antisymmetric in $a$, $b$ and $c$ vanishes for all $\epsilon$. Yet
these spinors are odd; for even spinors, we require the part of
$\langle \psi^a, (\epsilon \cdot \psi^b) \psi^c \rangle$ that is
\emph{symmetric} in $a$, $b$ and $c$ to vanish for all $\epsilon$.

Now let us bring in some division algebra technology to remove our
dependence on $\epsilon$. While we do this, let us replace $\psi^a$
with $\psi$, $\psi^b$ with $\phi$, and $\psi^c$ with $\chi$ to lessen
the clutter of indices.  Substituting in the formulas from
Table~\ref{tab:intertwiners}, we have
\begin{eqnarray*}
	\langle \psi, (\epsilon \cdot \phi) \chi \rangle & = & \Re( \psi^\dagger (\widetilde{\epsilon \phi^\dagger + \phi \epsilon^\dagger}) \chi) \\
	& = & \Retr(\psi^\dagger (\epsilon \phi^\dagger + \phi \epsilon^\dagger - \epsilon^\dagger \phi - \phi^\dagger \epsilon) \chi) \\
	& = & \langle \epsilon, (\psi \cdot \chi) \phi \rangle,
\end{eqnarray*}
where again we have employed the cyclic symmetry of the real trace, along with
the identity:
\[ \tr(\epsilon \phi^\dagger + \phi \epsilon^\dagger) = 
\Retr(\epsilon \phi^\dagger + \phi \epsilon^\dagger) = 
\phi^\dagger \epsilon + \epsilon^\dagger \phi. \]
This real quantity commutes and associates in any expression.
So, if we seek to show that the part of $\langle \psi, (\epsilon \cdot \phi)
\chi \rangle$ that is totally symmetric in $\psi$, $\phi$ and $\chi$ vanishes
for all $\epsilon$, it is equivalent to show the totally symmetric part of
$(\phi \cdot \chi) \psi$ vanishes. And since the dot operation in $\phi \cdot
\chi$ is symmetric, this follows immediately from our main result,
Theorem \ref{thm:cubic}.  
\end{proof}

\subsection*{Acknowledgements}  

We thank Geoffrey Dixon, Tevian Dray and Corinne Manogue for helpful
conversations and correspondence.  We also thank An Huang, Theo
Johnson-Freyd, Greg Egan, and David Speyer for catching some errors.  This work
was partially supported by an FQXi grant.

\end{document}